\documentclass[envcountsame,envcountsect]{llncs}
\usepackage{amsfonts,epsf,amsmath,amssymb}
\usepackage{enumerate}

\usepackage{tikz}
\usepackage{url}
\usepackage{fixltx2e}

\usepackage[hidelinks]{hyperref}
\usepackage{fge}

\spnewtheorem*{theoremA}{Theorem}
{\bf}{\itshape}
\spnewtheorem*{lemmaA}{Lemma}
{\bf}{\itshape}
\spnewtheorem*{propositionA}{Proposition}
{\bf}{\itshape}

\def\qed{\ifhmode\unskip\nobreak\fi\hfill
  \ifmmode\square\else$\square$\fi}

\newcommand{\NN}{\mathbb{N}}

\newcommand{\cclass}[1]{\mbox{{\sf #1}}}
\newcommand{\PP}{\cclass{P}}

\newcommand{\NPP}{\cclass{NP}}
\newcommand{\cNPP}{\cclass{co-NP}}

\newcommand{\cNPC}{\cclass{co-NP-complete}}

\newcommand{\cproblem}[1]{\textrm{\normalfont{#1}}}

\newcommand{\run}[3]{\cproblem{RUN\textsuperscript{#3}\textsubscript{(#1,#2)}}}
\newcommand{\runCD}{\run{$C$}{$D$}{}}

\newcommand{\runOD}{\run{1}{$D$}{}}
\newcommand{\runE}{\run{$\cdot$}{$\cdot$}{}}

\newcommand{\oo}{\operatorname{o}}
\newcommand{\OO}{\operatorname{O}}
\newcommand{\OOm}{\operatorname{\Omega}}
\newcommand{\TT}{\operatorname{\Theta}}

\newcommand{\ie}{i.e. }

\newcommand{\tup}[3]{#1, #2 \ldots #3}
\newcommand{\seq}[2]{#1, #2\ldots}

\newcommand{\crs}{\mathcal{C}}
\newcommand{\crsSet}{\mathcal{S}}

\newcommand{\blank}{\fgeupbracket}
\newcommand{\state}[1]{\textsc{#1}}
\newcommand{\qz}{q_{\state{0}}}
\newcommand{\qa}{q_{\state{acc}}}
\newcommand{\qr}{q_{\state{rej}}}

\title{Verifying whether One-Tape Non-Deterministic Turing Machines Run in Time $Cn+D$\thanks{This work is partially funded by the Slovenian Research Agency.}}
\author{David Gajser}
\institute{ IMFM, Jadranska 19, 1000 Ljubljana, Slovenija\\
\email{david.gajser@fmf.uni-lj.si}}

\begin{document}
\maketitle
\begin{abstract} We discuss the following family of problems, parameterized by integers $C\geq 2$ and $D\geq 1$: Does a given one-tape non-deterministic $q$-state Turing machine make at most $Cn+D$ steps on all computations on all inputs of length $n$, for all $n$?

Assuming a fixed tape and input alphabet, we show that these problems are \cNPC\ and we provide good non-deterministic and co-non-deterministic lower bounds. Specifically, these problems can not  be solved  in $\oo(q^{(C-1)/4})$ non-deterministic time by multi-tape Turing machines. We also show that the complements of these problems can be solved in $\OO(q^{C+2})$ non-deterministic time and not in $\oo(q^{(C-1)/2})$ non-deterministic time by multi-tape Turing machines.
\end{abstract}
\section{Introduction}
A year ago, there was a question on TCS Stack Exchange~\cite{forum} asking for a \emph{natural} problem in \NPP\ with a good deterministic time complexity lower bound. Everyone working in complexity theory knows, that such problems are hard to find. 
In the present paper we present a family of (relatively) natural \cNPC\ problems with good non-deterministic and co-non-deterministic lower bounds. Our results can be compared to the ones given by Adachi, Iwata and Kasai~\cite{Adachi} in 1984, where they proved good deterministic lower bounds for some problems that are complete in \PP.

Suppose someone gives us a 50-state one-tape non-deterministic Turing machine $M$. How can we verify whether $M$ runs in time $100n+1$? In other words, how can we verify, if $M$ makes at most $100n+1$ steps on each computation on each input of length $n$, for all $n$? Is there an algorithm that can help?

In the present paper we show how a Turing machine can verify whether a given $q$-state non-deterministic one-tape Turing machine \textbf{does not run} in time $100n+1$. We describe a non-deterministic multi-tape Turing machine that runs in time $\OO(q^{102})$ and does the job. What is more, we show that no $\OO(q^{49})$-time non-deterministic multi-tape Turing machine can do the same. We also show that there \textbf{does not} exist a non-deterministic multi-tape Turing machine that runs in time $\OO(q^{24})$ and solves the problem:
\begin{center}
\textbf{Does} a given $q$-state non-deterministic one-tape Turing machine run in time $100n+1$?
\end{center}
Moreover, this problem is shown to be \cNPC. All these results are presented in Sect.~\ref{uporaba}.

To prove the lower bounds, we make reductions from hard problems, for which hardness is proven by diagonalization. The diagonalization in Theorem~\ref{LB2} (non-deterministic lower bound) is straightforward and the diagonalization in Theorem~\ref{LB1} (co-non-deterministic lower bound) is implicit in the non-deterministic time hierarchy~\cite{Seiferas,Zak}. 
%
     
 For the upper bound, we use Theorem~\ref{glavni}, which we refer to as \emph{the compactness theorem}. Together with Lemma~\ref{lemma1} it assures, that in order to verify  whether a one-tape Turing machine runs in time $100n+1$, we only need to check the running time for small enough inputs. Section~\ref{glavno} is devoted to the compactness theorem.

We use crossing sequences to prove the compactness theorem. Specifically, we show that a Turing machine that runs in time $100n+1$ must produce some identical crossing sequences on each computation, if the input is long enough. Thus, when considering some fixed computation, we can partition the input on some parts where identical crossing sequences are generated and analyze each part independently. We prove that it is enough to consider small parts of the input. For more information and related work on crossing sequences, see \emph{Crossing Sequences} in Sect.~\ref{prelim}.

But why do we consider only one-tape Turing machines and linear time bounds, \ie the bounds $Cn+D$?  The answer is very simple: because essentially all other problems of verifying time bounds for Turing machines are un\-de\-ci\-da\-ble. This is argued in~\cite{jaz} where the author deals with the problem of verifying various time bounds for deterministic Turing machines. It is proven that, for multi-tape Turing machines, not even linear time bounds can algorithmically be verified. There is also no algorithm that would verify a time bound $T(n)=\OOm(n\log n)$, $T(n)\geq n+1$, for a given one-tape Turing machine. 
%
%
But if $T(n)=\oo(n\log n)$ is tangible enough, then there is an algorithm that verifies whether a given one-tape deterministic Turing machine runs in time $T(n)$. It is also shown in~\cite{jaz} that a one-tape Turing machine that runs in time $\oo(n\log n)$ must actually run in linear time, which implies that the most ``natural'' algorithmically verifyable time-bound for one-tape Turing machines is the linear one. Refering to~\cite{jaz} again, it is impossible to algorithmically verify whether a Turing machine (of any kind) runs in linear or even constant time, if the constant is not given. Thus we need $C$ and $D$ to be able to verify whether a one-tape Turing machine runs in time $Cn+D$.

%

\section{Preliminaries}
	\label{prelim}

\subsubsection{Basic Notation.} 
Let $\NN$ be the set of non-negative integers. 
All logarithms with no base written have base 2. We use $\epsilon$ for the empty word and $|w|$ for the length of a word $w$. For words $w_1$ and $w_2$, let $w_1w_2$ denote their concatenation.

We will use multi-tape Turing machines to solve decision problems. If not stated otherwise, lower and upper complexity bounds will be for this model of computation. We will not describe the model (any standard one will do). We will use notation DTM and NTM for deterministic and non-deterministic Turing machines.


\subsubsection{The Definition of Problems.} \emph{A one-tape NTM} is an 8-tuple $M=(Q, \Sigma, \Gamma, \blank,\\ \delta,\qz,\qa,\qr)$, where $Q$ is a set of states, $\Sigma\neq\emptyset$ an input alphabet, $\Gamma\supseteq\Sigma$ a tape alphabet, $\blank\in\Gamma\backslash \Sigma$ a blank symbol, $\delta:Q\backslash\{\qa,\qr\}\times\Gamma\rightarrow \mathcal{P}(Q\times\Gamma\times\{-1,1\})\backslash\{\emptyset\}$ a transition function and $\qz,\qa,\qr\in Q$ pairwise distinct starting, accepting and rejecting states. Here $\mathcal{P}$ denotes the power set.

As can be seen from the definition, the head of $M$ must move on each step and at the end of each finite computation the head of $M$ is in a halting state ($\qa$ or $\qr$). These properties will show very practical later, although the results of this paper do not rely on them.

For one-tape NTMs $M_1$ and $M_2$, the \emph{composition} of $M_1$ and $M_2$ is the NTM that starts computing as $M_1$, but has the starting state of $M_2$ instead of $M_1$'s accepting state. When the starting state of $M_2$ is reached, it computes as $M_2$. If $M_1$ rejects, it rejects.

\emph{A one-tape DTM} is a one-tape NTM where each possible configuration has at most one successive configuration.

The number of steps that a Turing machine $M$ makes on some computation $\zeta$ will be called \emph{the length of $\zeta$} and denoted by $|\zeta|$. 

For a function $T:\NN\rightarrow \NN$, if a Turing machine $M$, for each $n\in\NN$, makes \textbf{at most} $T(n)$ steps on all computations on inputs of length $n$, then we say that  $M$ \emph{runs in time $T(n)$}.

A main goal in this paper is to analyze the time complexity of problems 
$$\runCD=\{\textrm{one-tape NTMs that run in time }Cn+D\},$$ for $C,D\in\NN$. This will be done in Sect.~\ref{uporaba}, where we will also assume a fixed input alphabet $\Sigma$ and a fixed tape alphabet $\Gamma$. Strictly speaking, we will actually be analyzing problems $\runCD(\Sigma,\Gamma)$. This will enable us to have codes of $q$-state one-tape NTMs of length $\TT(q^2)$. 
Because $q$ will describe the length of the code up to a constant factor, we will usually express the complexity of algorithms with a $q$-state one-tape NTM as input in terms of $q$ instead of $n=\TT(q^2)$.

We use an overline to refer to the complements of the problems, like 
$$\overline{\runCD}=\{\textrm{one-tape NTMs that do not run in time }Cn+D\}.$$
The problem where $C$ and $D$ are parts of the input is denoted by $\runE$.


\subsubsection{Crossing Sequences.} For a one-tape Turing machine $M$, we can number the cells of its tape with integers so that the cell 0 is the one where $M$ starts its computation. Using this numbering we can number the boundaries between cells as shown on Fig.~\ref{tape}. Whenever we say that an input is written on the tape, we mean that its $i$th symbol is in cell $(i-1)$ and all other cells contain the blank symbol $\blank$.

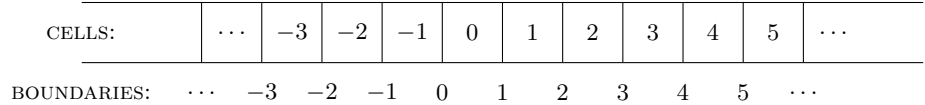
\begin{figure}[!htb]
	\begin{center}
	\begin{tikzpicture}[scale=0.8]
	\tikzstyle{prazno}=[inner sep=0pt, minimum size=0em]
		\node at ( -6,0) (c-50) [prazno] {};
		\node at ( -4,0) (c-40) [prazno] {};
		\node at ( -3,0) (c-30) [prazno] {};
		\node at (-2,0) (c-20)[prazno]  {};
		\node at ( -1,0) (c-10)[prazno]  {};
		\node at ( 0,0) (c00) [prazno] {};
		\node at ( 1,0) (c10) [prazno] {};
		\node at ( 2,0) (c20)  [prazno] {};
		\node at ( 3,0) (c30)[prazno]   {};
		\node at (4,0) (c40)[prazno]   {};
		\node at ( 5,0) (c50) [prazno]  {};
		\node at (6,0) (c60) [prazno]  {};
		\node at ( 8,0) [prazno]   {}
			edge [-] (c-50);
		\node at ( -6,-0.5) [prazno] {\sc{boundaries}:};
		\node at ( -4,-0.5)  [prazno] {$\ldots$ };
		\node at ( -3,-0.5)  [prazno] {$-3$};
		\node at (-2,-0.5) [prazno]   {$-2$};
		\node at ( -1,-0.5)[prazno]   {$-1$};
		\node at ( 0,-0.5)  [prazno] {$0$};
		\node at ( 1,-0.5) [prazno]  {$1$};
		\node at ( 2,-0.5)   [prazno] {$2$};
		\node at ( 3,-0.5) [prazno]   {$3$};
		\node at (4,-0.5)[prazno]   {$4$};
		\node at ( 5,-0.5)  [prazno]  {$5$};
		\node at (6,-0.5)  [prazno]  {$\ldots$};
		\node at ( 8,-0.5) [prazno]   {};
		\node at ( -6,1) (c-51) [prazno]  {};
		\node at ( -4,1) (c-41) [prazno]  {}
			edge [-] (c-40);
		\node at ( -3,1) (c-31) [prazno]  {}
			edge [-] (c-30);
		\node at (-2,1) (c-21)[prazno]   {}
			edge [-] (c-20);
		\node at ( -1,1) (c-11)  [prazno] {}
			edge [-] (c-10);
		\node at ( 0,1) (c01) [prazno] {}
			edge [-] (c00);
		\node at ( 1,1) (c11) [prazno]  {}
			edge [-] (c10);
		\node at ( 2,1) (c21) [prazno]  {}
			edge [-] (c20);
		\node at ( 3,1) (c31) [prazno]  {}
			edge [-] (c30);
		\node at (4,1) (c41)[prazno]  {}
			edge [-] (c40);
		\node at ( 5,1) (c51) [prazno]  {}
			edge [-] (c50);
		\node at (6,1) (c61) [prazno]  {}
			edge [-] (c60);
		\node at ( 8,1)  [prazno]  {}
			edge [-] (c-51);
		\node at ( -6,0.5) [prazno] {\sc{cells}:};
		\node at ( -3.5,0.5)  [prazno] {$\ldots$ };
		\node at ( -2.5,0.5)  [prazno] {$-3$};
		\node at (-1.5,0.5) [prazno]   {$-2$};
		\node at ( -0.5,0.5)[prazno]   {$-1$};
		\node at ( 0.5,0.5)  [prazno] {$0$};
		\node at ( 1.5,0.5) [prazno]  {$1$};
		\node at ( 2.5,0.5)   [prazno] {$2$};
		\node at ( 3.5,0.5) [prazno]   {$3$};
		\node at (4.5,0.5)[prazno]   {$4$};
		\node at ( 5.5,0.5)  [prazno]  {$5$};
		\node at (6.5,0.5)  [prazno]  {$\ldots$};
		\node at ( 8,0.5) [prazno]   {};
	\end{tikzpicture}
	\end{center}
	\caption{Numbering of tape cells and boundaries of a one-tape Turing machine.}
	\label{tape}
\end{figure}

Suppose that a one-tape NTM $M$ on first $t\in\NN\cup\{\infty\}$ steps of computation $\zeta$ on input $w$ crosses boundary $i$ of its tape at steps $\seq{t_1}{t_2}$ (this sequence can be finite or infinite). If $M$ was in state $q_j$ after the step $t_j$ for all $j$, then we say that $M$ produces the \emph{crossing sequence} $\crs^t_i(M,\zeta,w)=\seq{q_1}{q_2}$ and we denote its length by $|\crs^t_i(M,\zeta,w)|\in\NN\cup\{\infty\}$. Note that this sequence contains all information that the machine carries across the $i$th boundary of the tape in the first $t$ steps of computation $\zeta$. If we denote $\crs_i(M,\zeta,w)=\crs^{|\zeta|}_i(M,\zeta,w)$, the following trivial identity holds: $$|\zeta|=\sum_{i=-\infty}^\infty |\crs_i(M,\zeta,w)|.$$

Many properties of crossing sequences were proven already in 1960s by Hennie~\cite{hennie}, Hartmanis~\cite{Hartmanis} and Trakhtenbrot~\cite{Trakhtenbrot} and in 1980s by Kobayashi~\cite{Kobayashi}. They also proved that a one-tape deterministic Turing machine which runs in time $\oo(n \log n)$, produces only crossing sequences of bounded length and accepts a regular language. Later Tadaki, Yamakami and Lin~\cite{summary} proved the same for one-tape non-deterministic Turing machines. Properties of crossing sequences generated by one-tape non-deterministic Turing machines were also analyzed by Pighizzini~\cite{Pighizzini}.


\subsubsection{A Technical Lemma.} The proof of the next lemma can be found in Appendix~\ref{a1}.

\begin{lemma}
	\label{malaLema}
For every $q\geq 2$ and $C\in\NN$, it holds
$$\sum_{j=0}^C q^j(C-j)=\frac{q^{C+1}-(C+1)q+C}{(q-1)^2}\leq 4q^{C-1}.$$
\end{lemma}

\section{The Compactness Theorem}
	\label{glavno}

In this section we present the first result of this paper, the compactness theorem. If we explain it simple: If we want to verify that a NTM $M$ runs in time $Cn+D$, we only need to verify the number of steps that $M$ makes on inputs of some bounded length.
The same result for DTMs can be found in~\cite{jaz}, but the bound in the present paper is much better.

Before we formally state the theorem, let us introduce some notation. For a one-tape NTM $M$, define
$$\crsSet_n(M)= \{\crs_i^t(M,\zeta,w);\ |w|=n,\ 1\leq i \leq n,\ \zeta \textrm{ computation on input }w,\ t\leq |\zeta|\},$$
so $\crsSet_n(M)$ is the set of all possible beginnings of crossing sequences that $M$ produces on inputs of length $n$ on boundaries $\tup{1}{2}{n}$.

A bit more work will be with the definition of $t_M(w,\crs)$. Intuitively, $t_M(w,\crs)$ is the maximum number of steps that a one-tape NTM $M$ makes on a \textbf{part} $w$ of an imaginary input, if we only consider such computations, on which $M$ produces the crossing sequence $\crs$ on both two ending boundaries of $w$. To define it more formally, we will describe a valid \emph{computation of $M$ on part $w$ with ending crossing sequence $\crs=(\tup{q_1}{q_2}{q_l})$}. We will use the term \emph{standard case} to refer to the definition of computation of a NTM on a given input (not on a part).
Assume $|w|=n\geq 1$ and let $M=(Q,\Sigma,\Gamma,\blank, \delta,\qz,\qa,\qr)$.
%
\begin{itemize}
\item A valid configuration is a 5-tuple $(\crs_1,\tilde{w},i,\tilde{q},\crs_2)$, where $\crs_1$ is the \emph{left crossing sequence}, $\tilde{w}$ is some word from $\Gamma^n$, $0\leq i \leq n-1$ is the position of the head, $\tilde{q}\in Q$ is the current state of $M$ and $\crs_2$ is the \emph{right crossing sequence}. 
Intuitively, $\crs_1$ and $\crs_2$ are the endings of $\crs$ that still need to be matched.
\item The \emph{starting configuration} is $((\tup{q_2}{q_3}{q_l}),w,0,q_1,(\tup{q_1}{q_2}{q_l}))$. As in the standard case, we imagine the input being written on the tape of $M$ with the first bit in cell 0 (where also the head of $M$ is). The head will never leave the portion of the tape where input is written. Note that $q_1$ is missing in the left crossing sequence because we pretend that the head just moved from cell -1 to cell 0.
\item Valid configurations $A=(\crs_{1A},w_{A},i,q_A,\crs_{2A})$ and $B=(\crs_{1B},w_B,j,q_B,\crs_{2B})$ are successive, if one of the following holds:
	\begin{itemize}
	\item the transition function of $M$ allows $(w_A,i,q_A)$ to change into $(w_B,j,q_B)$ as in the standard case, $\crs_{1A}=\crs_{1B}$ and $\crs_{2A}=\crs_{2B}$,
	\item $i=j= 0$, $\crs_{1A}$ is of the form $(\tilde{q},q_B, \crs_{1B})$, $w_A=a\tilde{w}$, $w_B=b\tilde{w}$, $(\tilde{q},b,-1)\in\delta(q_A,a)$ and $\crs_{2A}=\crs_{2B}$,
	\item $i=j=n-1$, $\crs_{2A}$ is of the form $(\tilde{q},q_B, \crs_{2B})$, $w_A=\tilde{w}a$, $w_B=\tilde{w}b$ and $(\tilde{q},b,1)\in\delta(q_A,a)$ and $\crs_{1A}=\crs_{1B}$.
	\end{itemize}
\item There is a special \emph{ending configuration} that can be reached from configurations of the form
	\begin{itemize}
	\item $((q_l),a\tilde{w},0,\tilde{q},())$, if $(q_l,b,-1)\in\delta(\tilde{q},a)$ for some $b\in\Gamma$ or
	\item $((),\tilde{w}a,n-1,\tilde{q},(q_l))$, if  $(q_l,b,1)\in\delta(\tilde{q},a)$ for some $b\in\Gamma$.
	\end{itemize}
\item A valid computation of $M$ on part $w$ with ending crossing sequence $\crs$ is any sequence of successive configurations that begins with the starting configuration and ends with an ending configuration.
\end{itemize}
Similar as in the standard case, we can define $\crs_i(M,\zeta,w,\crs)$ as the crossing sequence, generated by $M$ on computation $\zeta$ on part $w\in\Sigma^n$ with ending crossing sequence $\crs$ on boundary $i$ ($1\leq i\leq n-1$) . We define
$$|\zeta|=\sum_{i=1}^{n-1}|\crs_i(M,\zeta,w,\crs)|+|\crs|$$
as the length of computation $\zeta$. To justify the definition, let us look at an example.

Suppose an input $w_1ww_2$ is given to $M$, $|w_1|,|w|\geq 1$. Let computation $\zeta_0$ produce the same crossing sequence $\crs$ on boundaries $|w_1|$ and $|w_1|+|w|$. Let $\zeta$ be the corresponding computation of $M$ on part $w$. Then $M$ on computation $\zeta_0$ spends exactly $|\zeta|$ steps on part $w$. What is more, if input $w_1w_2$ is given to $M$ (we cut out $w$) and we look at computation $\zeta_1$ which corresponds to $\zeta_0$, thus forming a crossing sequence $\crs$ on boundary $|w_1|$, then $|\zeta_1|=|\zeta_0|-|\zeta|$. Such considerations will be very useful in the proof of the compactness theorem.

We define $t_M(w,\crs)\in\NN\bigcup\{-1\}$ as the maximum length of computations of $M$ on part $w$ with ending crossing sequence $\crs$. If there is no valid computation of $M$ on part $w$ with ending crossing sequence $\crs$ or $|\crs|=\infty$, then we define $t_M(w,\crs)=-1$.

\begin{theorem}[The compactness theorem]
	\label{glavni}
Let $M$ be a one-tape NTM with $q$ states and let $C,D\in \NN$. 
 Denote $\ell=D+8q^C$, $r=D+12q^C$ and $\crsSet=\bigcup_{n=1}^\ell \crsSet_n(M)$. It holds:

$M$ runs in time $Cn+D$ \textbf{if and only if}
\begin{enumerate}[a)]
\item for each input $w$ of length at most $\ell$ and for each computation $\zeta$ of $M$ on $w$, it holds $|\zeta|\leq C|w|+D$ \textbf{and}
	\label{a}
\item for each $\crs\in \crsSet$ and for each part $w$ of length at most $r$, for which $t_M(w,\crs)\geq 0$,
 it holds $t_M(w,\crs)\leq C|w|$.
	\label{b}
\end{enumerate}
\end{theorem}
A skilled reader may notice that the following lemma already proves one direction of the theorem.
\begin{lemma}
	\label{lemma1}
Let everything be as in Theorem~\ref{glavni}. If \ref{b}) does not hold, then there exists some input $\tilde{w}$ of length at most $\ell+(Cr+D)r$ such that $M$ makes more than $C|\tilde{w}|+D$ steps on $\tilde{w}$ on some computation.
\end{lemma}
\begin{proof}
 If \ref{b}) does not hold, then there exists some finite crossing sequence $\crs\in \crsSet$, a part $w$ of length at most $r$ and a valid computation $\zeta$ of $M$ on part $w$ with ending crossing sequence $\crs$, such that $|\zeta|\geq C|w|+1$. From the definition of $\crsSet$ we know that there exist words $w_1$ and $w_2$ such that $|w_1|\geq 1$ and $|w_1|+|w_2|\leq \ell$, $t_0\in\NN$ and a computation $\zeta_0$, such that $\crs$ is generated by $M$ on input $w_1w_2$ on computation $\zeta_0$ on boundary $|w_1|$ after $t_0$ steps.

Now let us look at the input $\tilde{w}=w_1w^{Cr+D}w_2$ for $M$. Let $\tilde{\zeta}$ be a computation of $M$ on $\tilde{w}$, that on part $w_1$ (and left of it) and on part $w_2$ (and right of it) acts like the first $t_0$ steps of $\zeta_0$ and on parts $w$ it acts like $\zeta$. Note that after $\tilde{\zeta}$ spends $t_0$ steps on parts $w_1$ and $w_2$, crossing sequence $\crs$ is generated on boundaries $\tup{|w_1|}{|w_1|+|w|}{|w_1|+|w|^{Cr+D}}$. By that time, $M$ makes at least $t_0+(Cr+D)(C|w|+1)$ steps, which is at least $1+(Cr+D)(C|w|+1)\geq C(Cr+D)|w|+C(|w_1|+|w_2|)+D+1=C|\tilde{w}|+D+1$. Because $|\tilde{w}|\leq \ell+(Cr+D)r$, the lemma is proven.\qed
\end{proof}
Next, we prove the main lemma for the proof of the other direction of the compactness theorem.
\begin{lemma}
	\label{lemma2}
Let $C,D$ be non-negative integers, $M$ a one-tape $q$-state NTM and $w$ an input for $M$ of length $n$. Assume that $M$ makes at least $t\leq Cn+D$ steps on input $w$ on computation $\zeta$ and suppose that each crossing sequence produced by $M$ on $\zeta$ after $t$ steps on boundaries $\tup{1}{2}{n}$ appears at most $k$ times. Then $n\leq D+4kq^C$.
\end{lemma}
\begin{proof}
We know that $Cn+D\geq t\geq\sum_{i=1}^n|\crs_i^t(M,\zeta,w)|$, thus
\begin{align*}
n&\leq D+(C+1)n-\sum_{i=1}^n|\crs_i^t(M,\zeta,w)|
=D+\sum_{i=1}^n(C+1-|\crs_i^t(M,\zeta,w)|)\\
&\leq D+\sum_{j=0}^{C+1}\sum_{\substack{i=1\\ |\crs_i^t(M,\zeta,w)|=j}}^n(C+1- j)
\leq D+  \sum_{j=0}^{C+1}kq^j(C+1- j)
 \leq D+4kq^C,
\end{align*}
where the last inequality follows from Lemma~\ref{malaLema}.\qed
\end{proof}

Before going into the proof of the compactness theorem, let us define $w(i,j)$ as the subword of a word $w$, containing characters from $i$th to $j$th, including $i$th and excluding $j$th (we start counting with 0). Alternatively, if $w$ is written on a tape of a Turing machine, $w(i,j)$ is the word between the $i$th and $j$th boundary.
\begin{proof}[of the compactness theorem (Theorem~\ref{glavni})]
If $M$ runs in time $Cn+D$, then \ref{a}) obviously holds and \ref{b}) holds after Lemma~\ref{lemma1}. Now suppose that \ref{a}) and \ref{b}) hold. We will make a proof by contradiction, so suppose that $M$ does not run in time $Cn+D$. Let $w$ be the shortest input for $M$ such that there exists a computation of $M$ on $w$ of length more than $C|w|+D$. Denote this computation by $\zeta$ and let $n=|w|$, $t=Cn+D$. From \ref{a}) it follows that $n>\ell=D+4\cdot2q^C$, so after Lemma~\ref{lemma2} there exist at least three identical crossing sequences, produced by $M$ on input $w$ on computation $\zeta$ after $t$ steps on boundaries $\tup{1}{2}{n}$. Let these crossing sequences be generated on boundaries $ i_1< i_2<i_3$. Because $\crs_{i_1}^{t}(M,\zeta,w)$ and $\crs_{i_3}^{t}(M,\zeta,w)$ are of equal length, the head of $M$ is, before the $(t+1)$st step of computation $\zeta$, left of boundary $i_1$ or right of boundary $i_3$. Without loss of generality we can assume that the head is right from $i_3$ (if not, we can rename $i_1=i_2$ and $i_2=i_3$ and continue with the proof). Thus, no crossing sequence on boundaries $\tup{i_1}{i_1+1}{i_2}$ changes in the $(t+1)$st step of computation $\zeta$. Let $i_1\leq j_1<j_2\leq i_2$ be the closest boundaries such that $\crs_{j_1}^{t+1}(M,\zeta,w)=\crs_{j_2}^{t+1}(M,\zeta,w)$. Then crossing sequences $\crs_j^t(M,\zeta,w)$, for $j_1\leq j<j_2$, are pairwise distinct and do not change in the $(t+1)$st step of computation $\zeta$.

Let $\zeta_1$ be the computation on part $w(j_1,j_2)$ with ending crossing sequence $\crs$ that corresponds to $\zeta$ and let $\zeta_2$ be a computation on input $w(0,j_1)w(j_2,n)$ that in first $(t+1-|\zeta_1|)$ steps corresponds to the first $(t+1)$ steps of $\zeta$. Because input $w(0,j_1)w(j_2,n)$ is strictly shorter than $n$, $M$ makes at most $C(|w(0,j_1)|+|w(j_2,n)|)+D$ steps on any computation on this input, thus $t+1-|\zeta_1|\leq|\zeta_2|\leq C(|w(0,j_1)|+|w(j_2,n)|)+D$. From $t=Cn+D$ and $n= |w(0,j_1)|+|w(j_2,n)|+j_2-j_1$ it follows that $|\zeta_1|\geq C(j_2-j_1)+1$, thus $t_M(w(j_1,j_2),\crs)>C|w(j_1,j_2)|$.

Next, we will cut out some pieces of $w$ to eliminate as many redundant parts as possible (if they exist), while leaving the part of $w$ between boundaries $j_1$ and $j_2$ intact. Redundant parts are those, where identical crossing sequences are generated on computation $\zeta$ after $t$ steps. We will cut out parts recursively and the result will not necessarily be unique. 

Suppose that $\crs_k^{t}(M,\zeta,w)=\crs_l^{t}(M,\zeta,w)$ for $1\leq k<l\leq j_1$ or $j_2\leq k<l\leq n$. Cut out the part of $w$ between the $k$th and $l$th boundary.  Let $w'$ be the new input. Let the boundaries $j_1'$ and $j_2'$ for input $w'$ correspond to the boundaries $j_1$ and $j_2$ for input $w$. Let $\zeta'$ be a computation on $w'$ that corresponds to $\zeta$ (at least for the first $t$ steps of $\zeta$) and let $t'$ be the step in computation $\zeta'$ that corresponds to the step $t$ of computation $\zeta$. Now recursively find new $k$ and $l$. The recursion ends, where there are no $k,l$ to be found.

From the recursion it is clear that at the end we will get an input for $M$ of the form $w_0=w_1w(j_1,j_2)w_2$, where $|w_1|\geq 1$. Let $\zeta_0$ be the computation that corresponds to $\zeta$ after the cutting out (at least for the first $t$ steps of $\zeta$) and let $t_0$ be the step in $\zeta_0$ that corresponds to $t$. Denote $n_0=|w_0|$. From the construction it is clear that $M$ on input $w_0$ on computation $\zeta_0$ after $t_0$ steps generates the crossing sequence $\crs$ on boundaries $|w_1|$  and $|w_1|+j_2-j_1$. What is more, the crossing sequences on boundaries $\tup{1}{2}{|w_1|}$ are pairwise distinct. The same is true for the crossing sequences on boundaries $\tup{|w_1|+1}{|w_1|+2}{|w_1|+j_2-j_1}$ and the crossing sequences on boundaries $\tup{|w_1|+j_2-j_1}{|w_1|+j_2-j_1+1}{n_0}$.

Denote $\tilde{w}=w_1w_2$ and $\tilde{n}=|w_1|+|w_2|$. Let the computation $\tilde{\zeta}$ on $\tilde{w}$ be a computation that corresponds to $\zeta_0$ (at least for the first $t_0$ steps of $\zeta_0$) and let $\tilde{t}$ be the time step of $\tilde{\zeta}$ that corresponds to the time step $t_0$ of $\zeta_0$.
Because $\tilde{n}< n_0\leq n$ and because $w$ is the shortest input for $M$ that violates the $Cn+D$ bound, $M$ makes at most $C\tilde{n}+D$ steps on any computation on input $\tilde{w}$, thus also on computation $\tilde{\zeta}$. Note that no three crossing sequences from $\{\crs_i^{\tilde{t}}(M,\tilde{\zeta},\tilde{w});1\leq i\leq \tilde{n}\}$ are identical. Thus after Lemma~\ref{lemma2}, $\tilde{n}\leq D+4\cdot 2q^C=\ell$. Because $\crs_{|w_1|}^{\tilde{t}}(M,\tilde{\zeta},\tilde{w})=\crs$, it follows that $\crs\in\crsSet$. From \ref{b}) and $t_M(w(j_1,j_2),\crs)>C|w(j_1,j_2)|$ we can now deduce $j_2-j_1>r$.

There are at most three identical crossing sequences generated by $M$ on input $w_0$ on computation $\zeta_0$ after $t_0$ steps on boundaries $\tup{1}{2}{n_0}$. Because $n_0>j_2-j_1>r= D+4\cdot 3q^C$, it follows from Lemma~\ref{lemma2} that $t_0>Cn_0+D$. Because $n_0\leq n$ and $w$ is the shortest input on which the time bound $Cn+D$ is violated, $n_0=n$ and thus by the construction of $w_0$, $w_0=w$. It follows that $t_0=t$ which is a contradiction because $t=Cn+D$ and $t_0>Cn_0+D$.\qed


\end{proof}
\section{Computational Complexity of \runCD}
	\label{uporaba}
In this section we prove the main results of this paper.

Let us fix $\Sigma$ and $\Gamma$ for this section and assume $0,1\in\Sigma$. All Turing machines discussed in this sections will have the input alphabet $\Sigma$ and the tape alphabet $\Gamma$. 

Because we will deal with codes of one-tape NTMs a lot, let us say what properties we want from their encoding:
\begin{itemize}
\item given a code of a $q$-state one-tape NTM $M$, a multi-tape NTM can simulate each step of $M$ in $\OO(q^2)$ time,
\item the code of a  $q$-state one-tape NTM has to be of length $\TT(q^2)$ and has to start with at least one redundant zero, followed by a redundant 1 and
\item a code of a composition of one-tape NTMs can be computed in linear time by a multi-tape DTM.
\end{itemize}

An example of such an encoding is given in Appendix~\ref{kod}.
%

A \emph{padded code} of a one-tape NTM $M$ is any code of $M$, padded in front by any number of zeros. Thus the padded code of a one-tape NTM can be arbitrary long.

Now we are ready to state the first complexity bound. The detailed proof can be found in Appendix~\ref{dokazZG}, we give here just the idea.
\begin{proposition}
	\label{upper}
There exists a  multi-tape NTM that solves $\overline{\runE}$  in time $\OO(p(C,D)q^{C+2})$ for some quadratic polynomial $p$.
\end{proposition}
\begin{proof}[idea]
First, suppose that $C$ and $D$ are parameters and only a $q$-state one-tape NTM $M$ is given as input. To verify whether $M$ does not run in time $Cn+D$, we non-deterministically choose an input $w$ for $M$ (or its part) of length $\OO(q^C)$ and simulate $M$ on this input (or its part) to verify whether ~\ref{a}) (or~\ref{b})) from the compactness theorem fails. If so, we accept, else, we reject. Because we need $\OO(q^2)$ time to simulate one step of $M$ and we will not simulate more than $C|w|+D$ steps, we need $\OO(q^{C+2})$ steps overall. It can be shown that the constant hidden behind big $\OO$ notation is  $\OO(p(C,D))$, for some quadratic polynomial $p$.\qed
\end{proof}
%
%

The next two lemmas will play a crucial role in lower bound proofs we are about to give later. The detailed proofs can be found in Appendix~\ref{dokazaL}. The main idea is the same for both proofs, thus we only give the idea for the proof of Lemma~\ref{simulacija2}.
\begin{lemma}
	\label{simulacija}
Let $C\geq 2$ and $D\geq1$ be integers, let $T(n)= Kn^k+1$ for some integers $K,k\geq 1$ and let $M$ be a one-tape $q$-state NTM that runs in time $T(n)$. Then there exists an $\OO\left(\left(T(n)^{1/(C-1)}+n\right)^2\right)$-time multi-tape DTM that given an input $w$ for $M$, constructs a one-tape 
NTM $\tilde{M}$ such that
\begin{center}
$\tilde{M}$ runs in time $Cn+D$ \textbf{iff} $M$ rejects $w$.
\end{center}
\end{lemma}
\begin{lemma}
	\label{simulacija2}
Let $C\geq 2$ and $D\geq 1$  be integers and let $T(n)= Kn^k+1$ for some integers $K,k\geq 1$. Then there exists a multi-tape DTM $M_{\textrm{mult}}$, which given an input $(M,w)$, where $w$ is an input for a one-tape $q$-state NTM $M$, constructs a one-tape 
NTM $\tilde{M}$ such that
\begin{center}
$\tilde{M}$ runs in time $Cn+D$ \textbf{iff} $M$ makes at most $T(|w|)$ steps on input $w$.
\end{center}
We can make $M_{\textrm{mult}}$ run in time $\OO\left(\left(T(|w|)^{2/(C-1)}+|w|+q\right)^2+(|w|+q)^\kappa\right)$ for some integer $\kappa\geq1$, independent of $C,\ D,\ K$ and $k$.
\end{lemma}
\begin{proof}[idea] First, assume that $M$ is a parameter and only $w$ is the input. Let us describe $\tilde{M}$. 
The computation of $\tilde{M}$ on input $\tilde{w}$ consists of two phases. In the first phase, $\tilde{M}$ uses at most $(C-1)$ deterministic passes through the input to assure that $\tilde{w}$ is long enough, more specifically, $|\tilde{w}|=\OOm(T(|w|)^2)$. If  $\tilde{w}$ is shorter, $\tilde{M}$ rejects. In the second phase, $\tilde{M}$ computes $T(|w|)$ and simulates $M$ on $w$ for at most $T(|w|)$ steps.  If $M$ does not halt, $\tilde{M}$ starts an infinite loop, else it halts.

A straightforward implementation of the second phase can be done so that $\tilde{M}$ uses $\OO(T(|w|)^2)$ time in this phase, if the infinite loop does not occur (the square is because we need to count simulated steps). However, the main idea is in the first phase. It has to be implemented in such a way, that if the infinite loop does not occur, the second phase will take at most $|\tilde{w}|$ time, thus assuring that 
\begin{center}
$\tilde{M}$ runs in time $Cn+1\footnote{Why do we write $Cn+1$ and not just $Cn$? One reason is that all Turing machines make at least one step on the empty input.}$ \textbf{iff} $\tilde{M}$ runs in time $Cn+D$ \textbf{iff} $M$ makes at most $T(|w|)$ steps on input $w$.
\end{center}

The first phase needs to be implemented with as few states as possible to ensure that $M_{\textrm{mult}}$ will run in the appropriate time. The trick is to hard-code $(C-1)$ prime numbers of size $\TT(T(|w|)^{2/(C-1)})$ into the states of $\tilde{M}$. Then $\tilde{M}$ can make $(C-1)$ passes through the input in the first phase, each time verifying whether $|\tilde{w}|$ is divisible by some prime number. If the first phase does not reject, we know that $|\tilde{w}|=\OOm(T(|w|)^2)$, just as desired. 

Now suppose that $M$ is not a parameter any more. Considering that $\tilde{M}$ has $\OO\left(T(|w|)^{2/(C-1)}+|w|+q\right)$ states and that a NTM doing the second phase can be constructed in time $(|w|+q)^\kappa$, the result follows. \qed
\end{proof}
The next theorem puts \runCD\ into an appropriate complexity class.
\begin{theorem}
The problems $\runCD$ are \cNPC\ for all $C\geq 2$ and $D\geq 1$.
\end{theorem}
\begin{proof}
Proposition~\ref{upper} proves that these problems are in \cNPP\ and Lemma~\ref{simulacija} gives a Karp reduction of an arbitrary problem in \cNPP\ to the above ones. 
\qed
\end{proof}
The first lower bound follows.
\begin{theorem}
	\label{LB1}
Let $C\in\NN$ and $D\geq 1$ be constants. Then the problem $\overline{\runCD}$ can not be solved by a multi-tape NTM in time $\oo(q^{(C-1)/2})$.
\end{theorem}
\begin{proof}
For $C\leq 5$, the theorem holds (the length of the input is $\TT(q^2)$), so suppose $C\geq 6$.  By the non-deterministic time hierarchy theorem~\cite{Seiferas,Zak} there exists a language $L$ and a multi-tape NTM $M$ that decides $L$ and runs in time $\OO(n^{C-1})$, while no multi-tape NTM can decide $L$ in time $\oo(n^{C-1})$. We can reduce 
 the number of tapes of $M$ to get a one-tape NTM $M'$ that runs in time $\OO(n^{2(C-1)})$ and decides $L$.  By Lemma~\ref{simulacija}, there exists a multi-tape DTM $M_{\textrm{mult}}$ that runs in time $\OO(n^{4})$ and given an input $w$ for $M'$, constructs a one-tape $\tilde{q}$-state NTM $\tilde{M}$ such that
\begin{center}
$\tilde{M}$ runs in time $Cn+D$ \textbf{iff} $M'$ rejects $w$.
\end{center}
Because the description of $\tilde{M}$ has length $\OO(|w|^{4})$, it follows that $\tilde{q}=\OO(|w|^{2})$. If there was some multi-tape NTM that would decide $\overline{\runCD}$ in time $\oo(q^{(C-1)/2})$, we could combine it with $M_{\textrm{mult}}$ to get an $\oo(n^{C-1})$-time multi-tape NTM that would decide $L$, which does not exist. Hence, the problem $\overline{\runCD}$ can not be solved by a multi-tape NTM in time $\oo(q^{(C-1)/2})$. \qed
\end{proof}

\begin{theorem}
	\label{LB2}
Let $C\in\NN$ and $D\geq 1$ be constants. Then the problem $\runCD$ can not be solved by a multi-tape NTM in time $\oo(q^{(C-1)/4})$.
\end{theorem}
\begin{proof}
For $C\leq 9$, the theorem holds, so suppose $C\geq 10$. Let $\kappa$ be as in Lemma~\ref{simulacija2}.

 Let $M$ be the following one-tape NTM:
\begin{itemize}
\item On input $w$, which is a padded code of a one-tape NTM $M'$, construct a one-tape $\tilde{q}$-state NTM $\tilde{M}$ such that
\begin{center}
$\tilde{M}$ runs in time $Cn+D$ \textbf{iff} $M'$ makes at most $|w|^{\kappa(C-1)}$ steps on input $w$.
\end{center}
This can be done with a multi-tape DTM in time $\OO(|w|^{4\kappa})$ by Lemma~\ref{simulacija2}. Hence, it can be done in  $\OO(|w|^{8\kappa})$ by a one-tape DTM.
\item Verify whether $\tilde{M}$ runs in time $Cn+D$. If so, start an infinite loop, else halt.
\end{itemize}

Suppose that $\runCD$ can be solved by a multi-tape NTM in time $\oo(q^{(C-1)/4})$. Then it can be solved in time $\oo(q^{(C-1)/2})$ by a one-tape NTM. Considering $\tilde{q}^2=\OO(|w|^{4\kappa})$, we see that $M$ makes $\OO(|w|^{8\kappa})+\oo(|w|^{\kappa(C-1)})=\oo(|w|^{\kappa(C-1)})$ steps before starting the infinite loop.

It is easy to see (standard diagonalization) that $M$ on input $w$, which is a padded code of $M$, halts and makes more than $|w|^{\kappa(C-1)}$ steps. Since the padding can be arbitrary long, we have come to a contradiction.\qed
\end{proof}

\subsubsection{The Limits of Our Lower Bound Methods.} Let $T:\NN\rightarrow\NN$ be a function. In both our lower bounds (Theorem~\ref{LB1} and Theorem~\ref{LB2}) we used the following technique: We can simulate a Turing machine $M$ on input $w$ for $T(|w|)$ steps on a one-tape $\tilde{q}$-state NTM $\tilde{M}$ that runs in linear time, say $Cn+D$. We can do so by first ensuring that $\tilde{M}$ simulates $M$ only if the input $\tilde{w}$ for $\tilde{M}$ is long enough, specifically $\OOm(T(|w|))$. 

If we want $\tilde{M}$ to measure $T(|w|)$, it can not produce the same crossing sequence on some accepting computation on inputs of length $T(|w|)$. After Lemma~\ref{lemma2}, this implies $\tilde{q}^C=\OOm(T(|w|))$, hence by our method we can not prove e.g. a lower bound $\OOm(q^C)$ for solving $\runCD$ even if we use DTMs for solving it.

\subsubsection{An Open Problem.} \emph{For $D\in\NN$, how hard are the problems $\runOD$?}

 It is clear that we can solve problems $\run{C}{0}{}$, for $C\in\NN$,  in constant time. The answer is always NO, since any Turing machine makes at least one step on empty input.

It is also clear that we can solve problems $\run{0}{D}{}$, for $D\in\NN$, in polynomial time. We should only simulate a given one-tape NTM on inputs up to the length $(D+1)$. If they run in time $D$, they will never read the $(D+1)$st bit of the input.

For $C\geq 2$ and $D\geq 1$, good complexity bounds for $\runCD$ are given in this paper. Hence only the bounds for $C=1$ are missing.

For this case, it is easy to see that the problem $\run{1}{1}{}$ is solvable in  deterministic polynomial time. The reason for this is the fact, that a one-tape NTM that runs in time $(n+1)$ never moves its head to the left, except possibly in the first and the last step of the computation (the more detailed explanation of why this is so can be found in~\cite{jaz}).
Does a similar property hold for general $D$?

\subsubsection{Acknowledgements.} The author wishes to thank his research advisor Sergio Cabello for valuable comments.
\bibliographystyle{splncs03}
\bibliography{literature}

\newpage
\appendix
\section{Appendix}
\subsection{A Technical Proof from Section~\ref{prelim}}
	\label{a1}
Here we restate and prove Lemma~\ref{malaLema}.

\begin{lemmaA}[\ref{malaLema}]
For every $q\geq 2$ and $C\in\NN$, it holds
$$\sum_{j=0}^C q^j(C-j)=\frac{q^{C+1}-(C+1)q+C}{(q-1)^2}\leq 4q^{C-1}.$$
\end{lemmaA}
\begin{proof}
\begin{align*}
\sum_{j=0}^C q^j(C-j)&=C\sum_{j=0}^C q^j-q\frac{d}{dq}\left(\sum_{j=0}^C q^j\right)\\
&=C\frac{q^{C+1}-1}{q-1}-q\frac{d}{dq}\left(\frac{q^{C+1}-1}{q-1}\right)\\
&=\frac{q^{C+1}-(C+1)q+C}{(q-1)^2}.
\end{align*}
It is easy to see that, for $q\geq 2$, it follows $\frac{q^{C+1}-(C+1)q+C}{(q-1)^2}\leq\frac{q^{C+1}}{(q-1)^2}\leq 4q^{C-1}$.\qed
\end{proof}

\subsection{An Example of an Encoding of a One-Tape NTM}
	\label{kod}
In Sect.~\ref{uporaba}, we listed the following properties we want from the encoding of one-tape NTMs:
\begin{itemize}
\item given a code of a $q$-state one-tape NTM $M$, a multi-tape NTM can simulate each step of $M$ in $\OO(q^2)$ time,
\item the code of a  $q$-state one-tape NTM has to be of length $\TT(q^2)$ and has to start with at least one redundant zero, followed by a redundant 1 and
\item a code of a composition of one-tape NTMs can be computed in linear time by a multi-tape DTM.
\end{itemize}

Now we describe an example of such an encoding. A \emph{primitive code} of a $q$-state NTM $M$ is
\begin{itemize}
\item the code of a 3-dimensional array of size $q\times |\Gamma|\times q$, where the element with index $(q_1,a,q_2)$ is  a list of all pairs $(b,d)$, such that $(q_2,b,d)\in\delta(q_1,a)$ \textbf{or}
\item the code of a tuple $(\tup{M_1}{M_2}{M_k})$, where $M$ is a composition of one-tape NTMs $\tup{M_1}{M_2}{M_k}$.
\end{itemize}
Note that the same one-tape NTM can have several primitive codes (which is ok). Since some primitive codes can be too short, we assume that each primitive code is padded in front with $0^m1$, where $1\leq m\leq 4q^2$ ensures that the code is at least of length $q^2$. It follows that our code of a one-tape NTM always starts with redundant zeros, followed by a redundant one.

Note that all three desired properties hold for such an encoding. Specifically, if we want to compute the code of a $q$-state one-tape NTM $M$, which is a composition of several given one-tape NTMs, we can do so in $\OO(q^2)$ time by a multi-tape DTM.

\subsection{A Proof of Proposition~\ref{upper}}
	\label{dokazZG}

\begin{propositionA}[\ref{upper}]
There exists a  multi-tape NTM that solves $\overline{\runE}$  in time $\OO(p(C,D)q^{C+2})$ for some quadratic polynomial $p$.
\end{propositionA}
\begin{proof}
Le us describe a multi-tape NTM $M_\textrm{mult}$ that solves $\overline{\runE}$. 
\begin{itemize}
\item On input $(C,D,M)$, where $M$ is a $q$-state one-tape NTM, compute $\ell=D+8q^C$ and $r=D+12q^C$. 
\item Non-deterministically choose an input of length $n\leq1$ and simulate a non-deterministically chosen computation of $M$ on it. If $M$ makes more than $Cn+D$ steps, accept.
\item Non-deterministically choose words $w_1$, $w_2$ and $w_3$ such that $|w_1|\geq 1$, $1\leq|w_2|\leq r$ and $|w_1|+|w_3|\leq \ell$. Define empty crossing sequences $\crs_1$, $\crs_2$ and counters $t_0=C|w_0|+D$, $t_2=C|w_2|$.
\item Simulate a non-deterministically chosen computation $\zeta$ of $M$ on the input $w_0=w_1w_2w_3$. After each simulated step $t$  of $M$, do:
\begin{itemize}
	\item decrease $t_0$ by one,
	\item if the head of $M$ is on some cell $|w_1|\leq i<|w_1|+|w_2|$, decrease $t_2$ by one,
	\item update the crossing sequences $\crs_1=\crs_{|w_1|}^t(M,\zeta,w_0)$ and $\crs_2=\crs_{|w_1|+|w_2|}^t(M,\zeta,w_0)$.
	\item If $t_0<0$,  accept.
	\item Non-deterministically decide whether to do the following: \begin{itemize}\item If $\crs_1=\crs_2$ and $t_2<0$,  accept. Else, reject.\end{itemize}
	\item If $M$ halts,  reject.
\end{itemize}
Note that the counter $t_0$ counts the number of simulated steps, while the counter $t_2$ counts the number of steps done on the part $w_2$.
\end{itemize}

The compactness theorem assures that $M_\textrm{mult}$ correctly solves $\overline{\runE}$.

Because the condition $\crs_1=\crs_2$ is verified at most once during the algorithm and $|\crs_1|,|\crs_2|\leq C|w_0|+D\leq C(\ell+r)+D$, verification of $\crs_1=\crs_2$ contributes $\OO((CD+C+D+1)q^{C+1})$ time to the overall running time.
Because $M_{\textrm{mult}}$ needs $\OO(q^2)$ steps to simulate one step of $M$'s computation and it has to simulate at most  $C(\ell+r)+D$ steps, $M_{\textrm{mult}}$ runs in time $\OO((CD+C+D+1)q^{C+2})$.
\qed
\end{proof}

\subsection{Proofs of Lemmas~\ref{simulacija} and~\ref{simulacija2}}
	\label{dokazaL}
Before going into the proofs, let us give a simple corollary of the prime number theorem.
\begin{lemma}
	\label{prastevila}
For each integer $C\geq 1$, there exists a constant $N_C$ such that, for all $n\geq N_C$, there exist prime numbers $p_1,p_2\cdots p_{C}$, such that $n< p_1<p_2<\cdots<p_{C}\leq 2n$.
\end{lemma}
\begin{proof}
For $n\in\NN$, let $\pi(n)$ be the number of primes smaller than or equal to $n$. So $\pi(2n)-\pi(n)$ is the number of primes between $n$ and $2n$. The prime number theorem gives us 
$$\lim_{n\rightarrow\infty}\frac{\pi(n)\log(n)}{n}=1,$$
which implies
$$\lim_{n\rightarrow\infty}\frac{\pi(2n)-\pi(n)}{\frac{n}{\log n}}=1.
$$
Hence $\pi(2n)-\pi(n)>\frac{n}{2\log n}>C$, for all large enough $n$.\qed
\end{proof}

\begin{lemmaA}[\ref{simulacija}]
Let $C\geq 2$ and $D\geq1$ be integers, let $T(n)= Kn^k+1$ for some integers $K,k\geq 1$ and let $M$ be a one-tape $q$-state NTM that runs in time $T(n)$. Then there exists an $\OO\left(\left(T(n)^{1/(C-1)}+n\right)^2\right)$-time multi-tape DTM that given an input $w$ for $M$, constructs a one-tape 
NTM $\tilde{M}$ such that
\begin{center}
$\tilde{M}$ runs in time $Cn+D$ \textbf{iff} $M$ rejects $w$.
\end{center}
\end{lemmaA}
\begin{proof}
Let us first describe the NTM $\tilde{M}$. The computation of $\tilde{M}$ on input $\tilde{w}$ will consist of two phases. In the first phase, $\tilde{M}$ will 
use at most $(C-1)$ deterministic passes through the input to assure that $\tilde{w}$ is long enough. We will describe this phase in detail later. 



In the second phase, $\tilde{M}$ will write $w$ on its tape and simulate $M$ on $w$. Hence $\OO(|w|)$ states and $\OO(T(|w|))$ time are needed for this phase (note that $q$ is a constant). If $M$ accepts $w$, $\tilde{M}$ starts an infinite loop, else it halts. Let $c$ be a constant such that $\tilde{M}$ makes at most $cT(|w|)$ steps in the second phase before starting the infinite loop.


We make some preparations before explaining the first phase. Let $N_C$ be the constant from Lemma~\ref{prastevila} and let 
$$m=\max\left\{N_C,\left\lceil(cT(|w|))^{1/(C-1)}\right\rceil\right\}=\OO(T(|w|)^{1/ (C-1)}).$$
From Lemma~\ref{prastevila} it follows that there exist prime numbers $p_1,p_2\cdots p_{C-1}$  such that $m< p_1<p_2<\cdots<p_{C-1}\leq 2m$. 

Now we are ready to explain the first phase. 
The machine $\tilde{M}$ simply passes through the input $(C-1)$ times, each time verifying that $|\tilde{w}|$ is divisible by some $p_i$, for $i=\tup{1}{2}{(C-1)}$. If this is not the case, $\tilde{M}$ rejects. Else, the second phase is to be executed. It suffices to have $p_i$ states to verify in one pass if the input is divisible by $p_i$, so we can make $\tilde{M}$ have $\OO(\sum p_i)=\OO((C-1)m)=\OO(m)$ states for the first phase such that it makes at most $(C-1)|\tilde{w}|+1$ steps before entering the second phase. We assume that $\tilde{M}$ erases all symbols from the tape in the last pass of the first phase so that the second phase can begin with a blank tape.

If the second phase begins, we know that $|\tilde{w}|\geq p_1\cdot p_2 \cdots p_{C-1}> m^{C-1}\geq cT(|w|)$, thus $\tilde{M}$ makes at most $|\tilde{w}|$ steps in the second phase \textbf{iff} it does not go into an infinite loop. So we have proven that
\begin{center}
$\tilde{M}$ runs in time $Cn+1$ \textbf{iff} $\tilde{M}$ runs in time $Cn+D$ \textbf{iff} $M$ rejects $w$.
\end{center}

To construct $\tilde{M}$, we need to compute $m$ and find prime numbers $p_1,p_2\cdots p_{C-1}$, which takes at most $\OO(m^2+|w|)$ time: $T(|w|)$ can be computed in $\OO(|w|)$ time and we can test all numbers between $m$ and $2m$ for primality using AKS algorithm~\cite{AKS}. After that, the transition function on all $\OO(m+|w|)$ states of $\tilde{M}$ can  be straightforwardly constructed, so the description of $\tilde{M}$ can be obtained in $\OO((m+|w|)^2)$ time.\qed
\end{proof}
\begin{lemmaA}[\ref{simulacija2}]
Let $C\geq 2$ and $D\geq 1$  be integers and let $T(n)= Kn^k+1$ for some integers $K,k\geq 1$. Then there exists a multi-tape DTM $M_{\textrm{mult}}$, which given an input $(M,w)$, where $w$ is an input for a one-tape $q$-state NTM $M$, constructs a one-tape 
NTM $\tilde{M}$ such that
\begin{center}
$\tilde{M}$ runs in time $Cn+D$ \textbf{iff} $M$ makes at most $T(|w|)$ steps on input $w$.
\end{center}
We can make $M_{\textrm{mult}}$ run in time $\OO\left(\left(T(|w|)^{2/(C-1)}+|w|+q\right)^2+(|w|+q)^\kappa\right)$ for some integer $\kappa\geq1$, independent of $C,\ D,\ K$ and $k$.
\end{lemmaA}
\begin{proof}
The proof is very similar to the proof of Lemma~\ref{simulacija}. Again, we begin with the description of $\tilde{M}$. The computation of $\tilde{M}$ on input $\tilde{w}$ will consist of two phases. In the first phase, $\tilde{M}$ will measure the input and in the second phase, it will compute $T(|w|)$ and simulate $M$ on $w$ for at most $T(|w|)$ steps. If $M$ will not halt, $\tilde{M}$ will start an infinite loop, else it will halt. $\tilde{M}$ can do the second phase without the infinite loop  in time $\OO(T(|w|)^2)$ using $\OO(|w|+q)$ states.

A possible implementation of the second phase is the following:
\begin{itemize}
\item Compute $\textsl{counter}=T(|w|)$ and write its digits on cells with odd indices. 
\item Encode $[w$ with a marked first symbol$]$ on cells with even indices. The mark indicates the position of the head of $M$.
\item Simulate $M$ on $w$ on cells with even indices while for each simulated step reducing the \textsl{counter} by one.
\item If $M$ halts before the counter gets to zero, halt. Else, start an infinite loop.
\end{itemize}
Suppose $\tilde{M}$ makes at most $(cT(|w|))^2$ steps in the second phase before starting the infinite loop.

To describe the first phase of $\tilde{M}$'s computation, let $N_C$ be the constant from Lemma~\ref{prastevila} and let 
$$m=\max\left\{N_C,\left\lceil(cT(|w|))^{2/(C-1)}\right\rceil\right\}=\OO\left(T(|w|)^{2/(C-1)}\right).$$
The first phase of $\tilde{M}$'s computation is essentially the same as in the proof of Lemma~\ref{simulacija}. The machine $\tilde{M}$ simply passes through the input $(C-1)$ times, each time verifying, that $|\tilde{w}|$ is divisible by some prime $m<p\leq 2m$. If this is not the case, $\tilde{M}$ halts. Else, the second phase is to be executed.
$\tilde{M}$ uses $\OO(m)$ states for the first phase and it makes at most $(C-1)|\tilde{w}|+1$ steps before entering the second phase. $\tilde{M}$ can erase all symbols from the tape in the last pass of the first phase so that the second phase can begin with a blank tape.

If the second phase begins, we know that $|\tilde{w}|> m^{C-1}\geq (cT(|w|))^2$, thus $\tilde{M}$ makes at most $|\tilde{w}|$ steps in the second phase \textbf{iff} it does not go into an infinite loop. So we have proven that
\begin{center}
$\tilde{M}$ runs in time $Cn+1$ \textbf{iff} $\tilde{M}$ runs in time $Cn+D$ \textbf{iff} $M$ makes at most $T(|w|)$ steps on input $w$.
\end{center}

Now let us describe a multi-tape DTM $M_{\textrm{mult}}$ that constructs $\tilde{M}$ from $(M,w)$. Because $\tilde{M}$ uses $\OO(q+|w|)$ states for the second phase, $M_{\textrm{mult}}$ can construct a one-tape NTM $M_2$ that does the second phase in time $\OO((q+|w|)^\kappa)$ for some integer $\kappa$, independent of the function $T$. Independence can be achieved by considering that  $T(|w|)$ is computed at the beginning of the second phase with $\OO(1)$ states. Because the first phase does not depend on $M$, $M_{\textrm{mult}}$ can compute the DTM that does the first phase in time $\OO(m^2)$, as in the proof of Lemma~\ref{simulacija}.
 Since $\tilde{M}$ has $\OO(m+|w|+q)$ states and is a composition of $M_1$ and $M_2$, $M_{\textrm{mult}}$ can construct $\tilde{M}$ in time $\OO\left(\left(T(|w|)^{2/(C-1)}+|w|+q\right)^2+(|w|+q)^\kappa\right).$\qed
\end{proof}

\end{document}